\documentclass[leqno,letterpaper,10pt]{article}

\usepackage[margin=1in]{geometry}

\usepackage{amsmath,amssymb,amsthm}
\usepackage{hyperref}
\usepackage{booktabs}
\usepackage{lineno}
\modulolinenumbers[5]
\usepackage{tabularx}
\usepackage{authblk}
\usepackage{multirow}
\usepackage{pdflscape}

\usepackage{bm}
\usepackage{mathtools}
\usepackage[table]{xcolor}
\usepackage{xspace} 
\usepackage[normalem]{ulem}
\usepackage{multirow}
\usepackage{setspace}
\usepackage{tikz}
\usepackage{color}
\usepackage{fullpage}
\usepackage[capitalize]{cleveref}
\usepackage{subcaption}
\usepackage{pgf,tikz,pgfplots}
\usepackage[ruled]{algorithm}
\usepackage{algorithmic}
\pgfplotsset{compat=1.15}
\usepackage{mathrsfs}
\usetikzlibrary{arrows}
\crefname{problem}{problem}{problems}
\Crefname{problem}{Problem}{Problems}

\allowdisplaybreaks

\newtheorem{theorem}{Theorem}
\newtheorem{corollary}[theorem]{Corollary}
\newtheorem{lemma}[theorem]{Lemma}

\newtheorem{remark}{Remark}

\begin{document}

\newtheorem{modl}{ILP Formulation for $k$-Flow Decomposition}
\newenvironment{model}{\begin{samepage}\begin{modl}}{\end{modl}\end{samepage}}

\newcommand{\paths}{\mathcal{P}}
\newcommand{\subpaths}{\mathcal{R}}
\newcommand{\weights}{w}
\newcommand{\pathdecomp}{(\paths, \weights)}
\newcommand{\scs}{\mathcal{R}}
\newcommand{\fdsc}{(G, \scs)}
\newcommand{\val}{\textsc{val}}
\newcommand{\N}{\mathbb{N}}
\DeclarePairedDelimiter\ceil{\lceil}{\rceil}
\DeclarePairedDelimiter\floor{\lfloor}{\rfloor}

\newcommand{\red}[1]{{\color{red} #1)}}
\newcommand{\alexcom}[1]{{\color{blue} (Alex: #1)}}
\newcommand{\alex}[1]{{\color{blue} #1}}
\newcommand{\andi}[1]{{\color{orange} (Andi: #1)}}
\newcommand{\todo}[1]{{\color{blue} (TODO: #1)}}

\newcommand{\Gurobi}{\textsc{GUROBI}\xspace}
\newcommand{\ST}{\textsf{MFD-original}\xspace}
\newcommand{\OP}{\textsf{MFD-optimized}\xspace}
\newcommand{\IR}{$\mathsf{}$\xspace}
\newcommand{\SP}{\textsf{MFDSC-original}\xspace}
\newcommand{\SO}{\textsf{MFDSC-optimized}\xspace}
\newcommand{\IX}{\textsf{MIFD-original}\xspace}
\newcommand{\IO}{\textsf{MIFD-optimized}\xspace}
\newcommand{\TB}{$\mathsf{Toboggan}$\xspace}
\newcommand{\CS}{$\mathsf{Coaster}$\xspace}
\newcommand{\GR}{$\mathsf{Greedy}$\xspace}
\newcommand{\IA}{\textsf{MFDW}\xspace}
\newcommand{\CSH}{$\mathsf{Coaster Heuristic}$\xspace}
\newcommand{\CF}{$\mathsf{Catfish}$\xspace}
\newcommand{\IFD}{$\mathsf{IFDSolver}$\xspace}
\newcommand{\SA}{\textbf{SRR020730-Salmon}\xspace}
\newcommand{\StT}{\textbf{SRR307903-StringTie}\xspace}
\newcommand{\RS}{\textbf{Reference-Sim}\xspace}

\newcommand\relatedversion{}

\title{Accelerating ILP solvers for Minimum Flow Decompositions through search space and dimensionality reductions\thanks{This work was partially funded by the European Research Council (ERC) under the European Union's Horizon 2020 research and innovation programme (grant agreement No.~851093, SAFEBIO), and by the Academy of Finland (grants No.~322595, 352821, 346968).\\\indent We are very grateful to Manuel C\'aceres for very helpful discussing on sets of independent safe paths.}}
\date{}
\author[1]{Andreas Grigorjew}
\author[2]{Fernando H. C. Dias}
\author[3]{Andrea Cracco}
\author[3]{Romeo Rizzi}
\author[1]{Alexandru~I.~Tomescu}
\affil[1]{University of Helsinki}
\affil[2]{Aalto University}
\affil[3]{University of Verona}

\maketitle

\begin{abstract}
Given a flow network, the Minimum Flow Decomposition (MFD) problem is finding the smallest possible set of weighted paths whose superposition equals the flow. It is a classical, strongly NP-hard problem that is proven to be useful in RNA transcript assembly and applications outside of Bioinformatics.

We improve an existing ILP (Integer Linear Programming) model by Dias et al.~[RECOMB 2022] for DAGs by decreasing the solver's search space using \emph{solution safety} and several other optimizations. This results in a significant speedup compared to the original ILP, of up to 55-90$\times$ on average on the hardest instances. Moreover, we show that our optimizations apply also to MFD problem variants, resulting in similar speedups, going up to 123$\times$ on the hardest instances.

We also developed an ILP model of reduced dimensionality for an MFD variant in which the solution path weights are restricted to a given set. This model can find an optimal MFD solution for most instances, and overall, its accuracy significantly outperforms that of previous greedy algorithms while being up to an order of magnitude faster than our optimized ILP.
\end{abstract}


\thispagestyle{empty}
\pagebreak
\setcounter{page}{1}

\section{Introduction}
\label{sec:intro}

\paragraph{Motivation.} Minimum Flow Decomposition (MFD) is the problem of finding a minimum number of weighted paths to decompose a flow on a directed graph, such that the sum of weights of the paths crossing an edge is equal to the flow on that edge. 
Hartman et al.~\cite{hartman2012split} proved that MFD is NP-hard, even on directed acyclic graphs (DAGs), and even if the flow values are only in $\{1,2,4\}$.

MFD is a crucial tool in many applications, for example, in networking, where the paths represent traffic going through a sequence of routers. Minimizing the number of these paths reduces the amount of maintenance work, among other additional benefits~\cite{hartman2012split}. In RNA sequence reconstruction problems, 
MFD is used by e.g.~\cite{pertea2015stringtie,tomescu2013novel,gatter2019ryuto,bernard2014efficient,tomescu2015explaining,williams2019rna} to decompose a \textit{splice graph}, in which every node is a gene \emph{exon}, and every weighted path is an RNA sequence that needs to be reconstructed. This splice graph can, for example, be constructed by aligning RNA-seq reads to a reference genome.
MFD can also used to reconstruct viral quasispecies~\cite{baaijens2020strain}. The input graphs in these applications are DAGs, and all previous work on MFD that we cite below is also on DAGs. As such, in this paper, all graphs are DAGs.

Moreover, also MFD variants are of practical interest. For example, one practical strategy is to incorporate longer reads as \emph{subpaths} potentially spanning more than two nodes~\cite{pertea2015stringtie,shao2017accurate,williams2021flow,hagemann2020single}, which must appear in some solution path (i.e. reconstructed transcript or viral genome). Moreover, the edge weights may not form an exact flow, due to sequencing errors and read mapping artifacts. An MFD variant adapted to this case is to consider intervals of edge weights instead, called \emph{inexact flows} in~\cite{safikhani2013ssp,williams2019rna}.

\paragraph{Related work.} Given the complexity of the problem, it is common to use heuristic methods~\cite{shao2017theory}, for example, by using a greedy approach, iteratively removing the path with the largest currently available flow until the flow is fully decomposed. Cáceres et al.~\cite{reyes2022width} showed that for DAGs this approach performs well, with an approximation ratio of $O(\log \val(f))$ (where $\val(f)$ is the total flow of the graph), only if the \textit{width} of the graph (the minimum number of paths to cover all edges) does not increase in the process; otherwise it gives an exponentially worse result than the optimal with an approximation ratio of $\Omega(m/\log m)$. 

Mumey et al.~\cite{mumey2015parity} gave an approximation algorithm on DAGs with ratio $\lambda^{\log\lVert f\rVert}\log\lVert f\rVert$, where $\lambda$ is the longest source-to-sink path length and $\lVert f\rVert$ is the maximum-norm on the flow $f$ (i.e., the largest flow value of the graph). Polynomial-time approximation algorithms with sublinear ratio in the size of the graph are not known, neither whether the problem is in APX, i.e.~admits a constant-factor approximation algorithm.

Kloster et al.~\cite{kloster2018practical} proposed the first algorithm solving MFD optimally, which is linear-time fixed-parameter tractable (FPT), parameterized on the MFD size $k$. While it performs well for small $k$ values, it quickly becomes unfeasible for larger $k$ values, since the parameterized runtime grows exponentially on $k$ with a degree $2$ polynomial in the exponent. This is a limiting factor to the usability of such algorithms in practical applications. 

In a recent work, Dias et al.~\cite{dias2022fast} formulated an Integer Linear Problem (ILP) model for MFD on DAGs and showed that it performs faster on graphs with larger solution sizes than the FPT algorithm of Kloster et al.~\cite {kloster2018practical}. Due to this better performance of an ILP, it is a natural to work on further optimizations. An additional strength of ILP formulations is their extensibility, since ILPs can easily be modified to handle other aspects of the applications, such as subpath constraints and inexact flows cited above, see~\cite{dias2022fast}.

\paragraph{Contributions.} In this paper, we significantly optimize the ILP for MFD using the notion of \emph{safe paths} for all flow decompositions in DAGs, first studied by Khan et al.~\cite{Khan:2022wo}. More specifically, a path $P$ is said to be \emph{safe} if all flow decompositions (of any size) have at least one path containing $P$ as a subpath. Clearly, safe paths must appear in any MFD. We show how they can be used to reduce the ILP search space by using them to fix many binary variables of the ILP. At a high level, this approach is along the same lines as for other NP-hard problems, for example, Bumpus et al.~\cite{bumpus2022search} analyzed \textit{$c$-essential vertices}, as those vertices belonging to any $c$-approximate solution, for some specific graph problems whose solutions are sets of vertices. Clearly, $c$-essential vertices also belong to all optimal solutions, and can be simply removed from the graph when running e.g.~an FPT algorithm (see~\cite{bumpus2022search} for details). 

In our case, safe solutions are sets of \emph{paths} (not single vertices), and it is non-trivial how to use them to simplify the input graph, or how to integrate them into a combinatorial algorithm (e.g., we cannot simply remove them, they could overlap, etc.). 
However, incorporating safe paths into an ILP is easily supported since they can be modelled by additional constraints, or by fixing some ILP variables. Moreover, in order to use more than one safe path, we observe that pairwise unreachable paths (they are not subpaths of a superpath) must be present in \emph{different} MFD solution paths. As a fast heuristic to select a set of such paths of large total length, we use a reduction to a maximum weight antichain problem~\cite{rival2012graphs}. In \Cref{apx:optimal-safe-paths} we show that a \emph{maximum-length} set of pairwise unreachable safe paths can also be found in polynomial time, but this procedure is overall computationally more involved, and thus may be too expensive as a pre-processing step.


Furthermore, as observed by C\'aceres et al.~\cite{reyes2022width}, the size of any edge antichain (like ours) is a lower bound on the MFD size. As such, we can also use this lower bound to check if a heuristic MFD solution size attains this lower bound, in which case it is optimal. To the best of our knowledge, this is the first time such a (non-trivial) lower bound is used in a solver.
Lastly, before running the ILP solver, we also simplify the input graph using the Y-to-V graph compaction, also used in the implementation of the FPT algorithm of Kloster et al.~\cite {kloster2018practical}, but not used in the one of the ILP by Dias et al.~\cite{dias2022fast}. Using all these optimizations, we experimentally show that we obtain significant speedups over the original ILP, of up to 55-90$\times$ on the harder inputs. 

In addition, we also adapt these optimizations to the MFD variants considering inexact flows and subpath constraints, proving that such adaptations can easily be done also for MFD variants of practical interest. Furthermore, we also show (in \Cref{sec:y-to-v}) how the Y-to-V graph compaction can be modified for these two variants, an issue which has not been considered before in the literature. With these optimizations, on the hardest graphs among the instances with inexact flows, we obtain an impressive speedup of 123$\times$ compared to the original ILP for MFD with subpath constraints. 

As a last contribution, we tackle an MFD variant where the weights of the solution paths must belong to a given set. We show that this problem on DAGs admits a simpler and thus faster ILP than the one for MFD. When the set of given weights consists of powers of 2 (up to the maximum flow value of any edge), this ILP provides a $\log||f||$-factor approximation ratio for the original MFD problem. For practical applications, we show that if we also add the flow values of all edges to the given set of weights, then this ILP returns an MFD solution with a number of paths that is significantly closer to the minimum one than the state-of-the-art heuristic MFD solver \CF~\cite{shao2017theory}, on instances with large MFD size. These are of particular interest, since our optimized solver runs in under 2 seconds on instances of solution size at most $10$. The ILP for this problem variant is up to a further order of magnitude faster than our optimized ILP, and running in under 1 second on average. 

\section{Preliminaries}
\label{sec:basic}

\paragraph{Definitions.}
We let $G=(V,E)$ be a graph and assume throughout the paper that $G$ is a DAG (i.e. directed and acyclic). For simplicity we also assume that $G$ has no parallel edges and that it has a single source node $s$ and a single sink node $t$.
We denote by $\N$ the natural numbers including $0$ and by $\mathbb{Z}^+ = \mathbb{N}\setminus\{0\}$ the positive integers.
A \emph{$u$-$v$ path} $P$ is a sequence of edges going from the node $u$ to $v$ and its length $|P|$ is defined as its number of edges. Additionally, we identify paths $P$ with functions $P:E\to\{0,1\}$ where $P(e) = 1$ iff $e \in P$. We call two paths $P_1, P_2$ (and resp. edges) \emph{independent}, if there exists no path $P^*$ such that $P_1$ and $P_2$ are subpaths of $P^*$. For a path $P=(e_1,\dots,e_{|P|})$ we denote by $P[\ell..r]\coloneqq(e_\ell,\dots,e_r)$ a subpath of length $r-\ell+1$. A \emph{flow network} is the tuple $G = (V,E,f)$, where $f:E\to\mathbb{Z}^+$ is a flow, i.e. a function that preserves the flow conservation
$\sum_{(u,v) \in E} f_{uv} = \sum_{(v,w) \in E} f_{vw}$, $\forall v \in V \setminus \{s,t\}$.

By a \emph{$k$-flow decomposition} of a flow $f$ (see \Cref{fig:FD}) we denote a family of $s$-$t$ paths $\paths = (P_1,\dots,P_k)$ with their associated weights $w=(w_1,\dots,w_k)\in(\mathbb{Z}^+)^k$, such that:
\begin{equation}
\label{eqn:superposition_of_flow}
\sum_{i \in \{1,\dots,k\}} P_i(u,v)w_i = f_{uv}, \quad \forall (u,v) \in E.
\end{equation}
In this paper, we focus on positive integer flow values, which is motivated by its application in RNA sequence reconstruction. Vatinlen et al.~\cite{VATINLEN20081390} proved that there is always a $(|E|-|V|+2)$-flow decomposition. The \emph{Minimum Flow Decomposition (MFD)} problem asks for a $k$-flow decomposition with minimum $k$.

In case of imperfect data, Williams et al.~\cite{williams2019rna} considered the problem variant \emph{Minimum Inexact Flow Decomposition (MIFD)} where for each edge $(u,v)$ we have a lower bound $L_{uv}$ and an upper bound $R_{uv}$, and constraint \eqref{eqn:superposition_of_flow} is changed to require that $\sum_{i \in \{1,\dots,k\}} P_i(u,v)w_i \in [L_{uv},R_{uv}]$, $\forall(u,v) \in E$. Another MFD variant of practical interest is \emph{Minimum Flow Decomposition with Subpath Constraints (MFDSC)}~\cite{williams2021flow}. The only difference with respect to MFD is that in the input we also have a set of paths $\mathcal{R}$, called \emph{subpath constraints}, and it is required that every path in $\mathcal{R}$ is a subpath of at least one path of the 
decomposition.



\begin{figure}[!h]
\centering
\begin{subfigure}[c]{0.49\linewidth}
\centering
\includegraphics[width=0.65\textwidth]{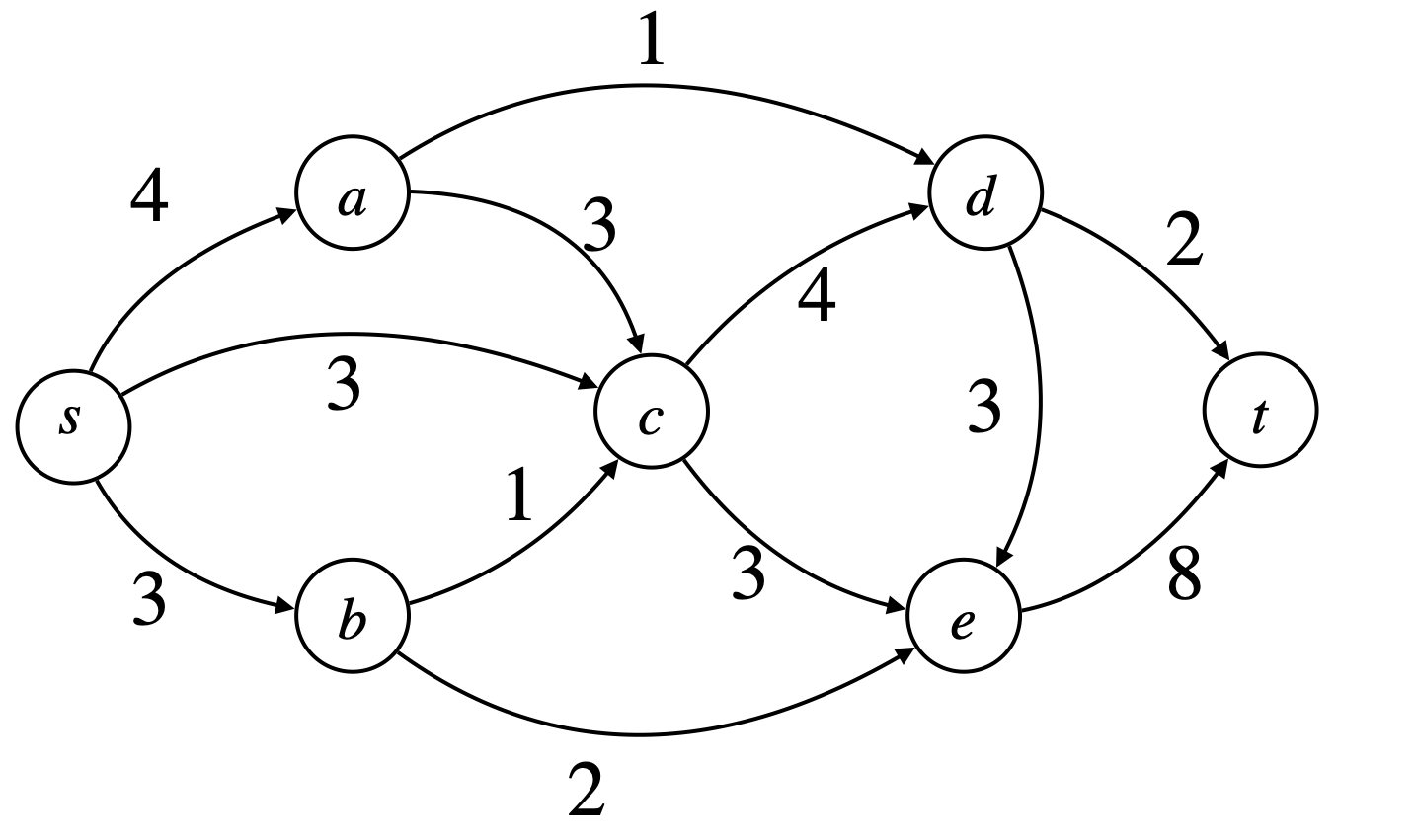}
\caption{A flow network \label{fig:MFD}}
\end{subfigure}
\begin{subfigure}[c]{0.49\linewidth}
\centering
\includegraphics[width=0.65\textwidth]{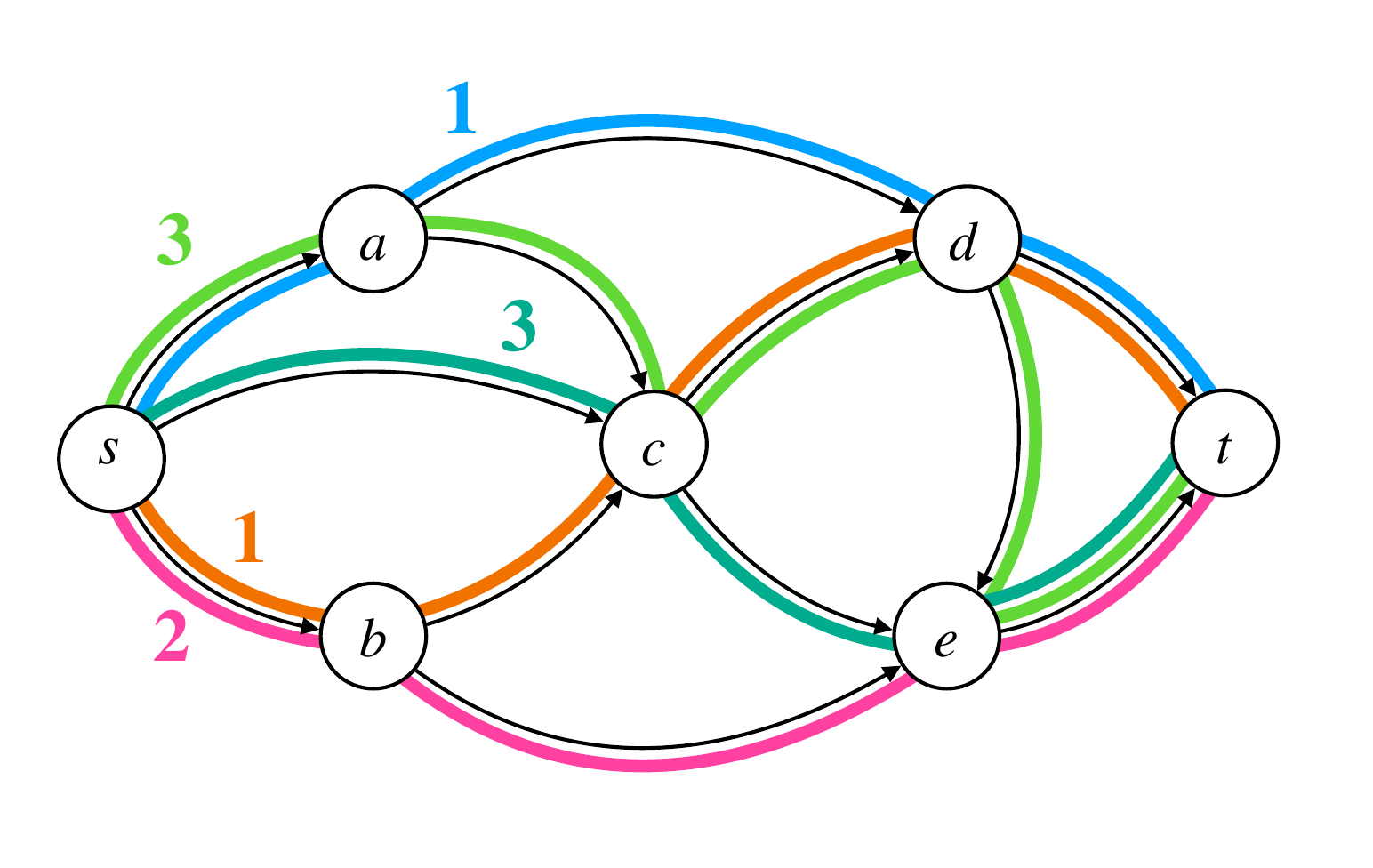}
\caption{A flow decomposition into paths of weights $(1,1,2,3,3)$ 
\label{fig:3-FD}}
\end{subfigure}
\caption{Example of a flow network in \eqref{fig:MFD} and a decomposition of it into five $s$-$t$ paths in \eqref{fig:3-FD}\label{fig:FD}}
\end{figure}

\paragraph{ILP formulations.}
We now review the ILP formulation for MFD given by Dias et al.~\cite{dias2022fast}. To model each path $P_i$, $i \in \{1,\dots,k\}$, binary variables $x_{uvi}$ are introduced to represent each edge $(u,v) \in E$. We set $x_{uvi}=P_i(u,v)$, i.e.~$x_{uvi} = 1$ if $(u,v)\in P_i$ and otherwise $x_{uvi} = 0$. At the same time, each path is required to start from the source and end in the sink. For all the intermediary nodes of a path, a unit in-degree and a unit out-degree is required. Those requirements can be modeled by the following constraints. For all $v \in V$ and for all $i \in \{1,\dots,k\}$:

\begin{equation}
\label{eqn:path_definition}
    \sum_{(u,v) \in E} x_{uvi} - \sum_{(v,u) \in E} x_{vui} = 
    \begin{cases}
    0, & \text{if $v \in V \setminus \{s,t\}$}, \\
    1, & \text{if $v = t$}, \\
    -1, & \text{if $v = s$},
    \end{cases} \quad\quad
\end{equation}

Constraint \eqref{eqn:superposition_of_flow} is then modeled as:
\begin{equation}
\label{eq:ILP-flow-superposition}
 f_{uv} - \sum_{i \in \{1,\dots,k\}} x_{uvi}w_i = 0,\quad\forall (u,v) \in E.
\end{equation}

The above constraint contains the non-linear term $x_{uvi}w_i$ that standard linear solvers cannot solve. However, using basic linearization techniques, each such constraint can be replaced by three linear constraints, see \cite{dias2022fast} for details. For completeness, we give the full formulation in \Cref{apx:standard-ilp}.

Dias et al.~\cite{dias2022fast} used the flexibility of this ILP formulation to easily model also the two variants MIFD and MFDSC. The former problem can be trivially modeled in ILP by changing the (linearized version of) constraint \eqref{eq:ILP-flow-superposition} to state that $\sum\limits_{i \in \{1,\dots,k\}} x_{uvi}w_i$ belongs to the interval $[L_{uv},R_{uv}]$. The latter problem can be modeled by introducing additional indicator variables $r_{ij} \in \{0, 1\}$ for every $i \in \{1,\dots,k\}$, and every path $R_j \in \mathcal{R}$, modeling if the subpath $j$ is contained in the $i$-th path of the flow decomposition, via the additional constraints:
\begin{align}
    \sum_{(u,v)\in R_j} x_{uvi} \geq |R_j|r_{ij},\quad &\forall i\in\{1,\dots,k\}, \forall R_j\in\subpaths,\label{eq:path-contains-subpath} \\
    \sum_{i\in\{1,\dots,k\}} r_{ij} \geq 1,\quad & \forall R_j\in\subpaths. \label{eq:subpath-in-atleast-one-path}
\end{align}




\section{Reducing the ILP search space via safe paths}
\label{sec:optimized-ilp}

In this section, we optimize the ILP model for MFD presented above based on the notion of \emph{solution safety} (see e.g.~Tomescu and Medvedev~\cite{tomescu2017safe}). Then, in \Cref{rem:safety-for-variants,rem:safety-for-subpaths}, we discuss how this method can also be applied to the MIFD and MFDSC problems.

We call a path \textit{safe} if it is part of every (not necessarily optimal) flow decomposition; that is, if it is a subpath of some path in every decomposition (not necessarily of minimum size). Khan et al.~\cite{Khan:2022wo} characterized all safe paths via the notion of \emph{excess flow}, which can be thought of as the flow value of the first edge of the path minus the flow values of the non-path edges out-going from the internal nodes of the path (i.e., the flow ``leaking'' from the path). Formally:


\begin{lemma}[\cite{Khan:2022wo}]
For a path $P = (e_1,\dots,e_{|P|})$ with $e_i=(u_i,u_{i+1})$, define its \emph{excess flow} as
\[f_P\coloneqq f(e_1) - \sum_{\substack{(u_i,v)\in E\\i\in\{2,\dots,|P|-1\},v\neq u_{i+1}}} f(u_i,v).\]
A path $P$ is safe if and only if $f_P > 0$.
\end{lemma}
Since we assume the flow to be positive, every edge (as a path of length $1$) is safe. Khan et al.~\cite{Khan:2022wo} gave a simple algorithm for finding all safe paths, as follows. We can find all safe subpaths of each $s$-$t$ path $P=(e_1,\dots,e_{|P|})$ of an arbitrary flow decomposition with a two-pointer algorithm. Starting with $\ell = 1$ and $r = 2$ as an inclusive interval (i.e., including two edges), we increase $r$ as long as $f_{P[\ell..r]} > 0$, and increase $\ell$ as long as $f_{P[\ell..r]} \leq 0$. That is, for every index $r=2,\dots,|P|$, we find the minimum $\ell<r$ such that the subpath $P[\ell..r]$ is safe. The runtime of this procedure is $O(\text{out-degree}(P))$, where $\text{out-degree}(P)\coloneqq\sum_{(u,v)\in P}\text{out-degree}(u)$, since flows of edges across $P$ are at most added and subtracted twice throughout the algorithm. To find the set $\mathcal{S}$ of all safe paths of $G$, we can find any flow decomposition $\paths$ and run the two-pointer algorithm on all paths in $\paths$. For quickly obtaining a flow decomposition, one could e.g. use the greedy approach of removing paths of largest currently available flow in runtime $O(|\paths|\cdot(n+m))$ and use it as upper bound to the optimal solution in the ILP~\cite{VATINLEN20081390}. In our implementations for the three problems MFD, MIFD, MFDSC we use flow decompositons obtained from the fastest heuristic solvers for these problems.

An \emph{(edge) antichain} is a subset $Q\subseteq E$ such that every pair of edges $e_1, e_2\in Q$ is independent. The size of any edge antichain is also a lower bound on the size of an MFD, as noted by~\cite{reyes2022width}, since each edge in the antichain must be traversed by \emph{different} paths in any flow decomposition. The idea is to use such edges (and extensions of them via safe paths, as we discuss next) to fix some ILP variables.

For every edge $e \in E$, denote by $a(e)$ the length of the longest safe path traversing $e$, i.e.~$a(e) \coloneqq \max_{P \in \mathcal{S}, e \in P} |P|$; this can be obtained from the computation of safe paths described above. Let $\paths(Q) = \{P_1,\dots,P_{|Q|}\}$ be the corresponding longest safe paths passing through the edges in $Q$; see \Cref{fig:safe-paths} in \Cref{apx:additional-proofs} for an illustration. If $Q \subseteq E$ is an antichain, then also the corresponding longest safe paths for the edges of $Q$ must be fully traversed by \emph{different} paths in any flow decomposition (otherwise, if some $P_i$ and $P_j$ were traversed by the same $s$-$t$ path in some flow decomposition, then also the corresponding edges in $Q$ could be traversed by the same $s$-$t$, contradicting the fact that $Q$ is an antichain).

As such, without loss of generality, we can fix in the ILP the $i$-th path of the MFD to contain the $i$-th safe path in $\paths(Q)$:
\begin{equation}
\label{eq:safety-constraints}
    x_{uvi} = 1,\quad \forall (u,v)\in P_i, \forall P_i\in \paths(Q).
\end{equation}

Additionally, in order to further optimize the choice of the antichain $Q$, we consider the following notion. Given a weight function $a:E\to\mathbb{N}$, a \emph{maximum weight antichain} is an antichain $Q$ maximizing $\sum_{e\in Q}a(e)$~\cite{rival2012graphs}. Maximum weight antichains can be found by a reduction to minimum flows, with demands on the edges given by their weights, in runtime $O(m\cdot\sum_{e\in E}a(e))\subseteq O(n\cdot m^2)$, followed by a depth-first search through the graph. Maximum weight antichains are cut-sets with edges whose minimum flow is equal to their demand~\cite{rival2012graphs}. In order to heuristically fix as many $x_{uvi}$ variables as possible, we set the weight $a(e)$ of each edge $e$ to be the length of the longest safe path traversing $e$. The overall runtime of this preprocessing is thus $O(m\cdot\sum_{e\in E}a(e))$. Note that this approach does not yield the maximum total length of independent safe paths, we describe an approach on how to find them in \Cref{apx:optimal-safe-paths}.

As further optimization, note that symmetries are well known to slow down ILP solvers. In the previous ILP formulations, the paths can be arbitrarily permuted, which we can mitigate with the following constraints: $w_{i+1} \leq w_i, \forall i \in \{|Q|+1,\dots,k-1\}.$

Note that we exclude the first $|Q|$ weights as they have already been fixed to use the longest safe paths crossing the antichain $Q$. This approach does not introduce more variables but adds constraints. Moreover, it can be effective in case the paths are forced to have widely different weights, which is the case in practical inputs for the RNA reconstruction problem~\cite{li2014rnaseqreadsimulator}.
We show the complete optimized ILP in \Cref{apx:optimized-ilp}.

Moreover, in our optimized ILPs we have an additional check on whether the lower bound $|Q|$ equals the size of the heuristic flow decomposition (which we need in any case to initially compute the safe paths, as mentioned above). If this holds, then the heuristic flow decomposition has minimum size, and we directly report it, without running the ILP.

\begin{remark}
\label{rem:safety-for-subpaths}
Since any flow decomposition satisfying the subpath constraints is (trivially) a flow decomposition, it follows that safe paths for all flow decompositions are also safe for flow decompositions with subpath constraints. However, as for the inexact problem variant, they do not capture \emph{all} the safe paths for MFDSC. Nevertheless, we can still use such paths with positive excess flow as described above to again reduce the search space of the ILP solver for MFDSC.
\end{remark}

\begin{remark}
\label{rem:safety-for-variants}
Since in the MIFD problem we are given intervals, instead of a single value, for every edge, the excess flow characterization needs to be adapted. In fact, it is an open problem how to find all the safe paths for all feasible inexact flow decompositions (of any size). However, in this paper we consider a ``conservative'' adaptation of the excess flow notion (\emph{inexact excess flow}), and prove that paths having positive inexact flow are safe for all inexact flow decompositions. As such, we can use them in the same manner as described above to reduce the search space of the ILP solver for MIFD.

\begin{lemma}[Inexact excess flow]
    \label{lem:inexact-excess-flow}
    For a path $P = (e_1,\dots,e_{|P|})$ with $e_i = (u_i, u_{i+1})$ define its \emph{inexact excess flow} as
    \[ f_P \coloneqq L_{e_1} - \sum_{\substack{(u_i, v)\in E\\i\in\{2,\dots,|P|-1\},v\neq u_{i+1}}} R_{u_iv}. \]
    A path $P$ is safe if $f_P > 0$.
\end{lemma}
Note that we lose the \textit{if and only if} property by assuming the worst case of only sending the lowest possible amount of flow through the first edge and of removing the largest possible amount of flow through the outgoing edges of the path. This means that inexact excess flows will not necessarily find all safe paths.

\end{remark}

\section{Y-to-V reduction for subpath constraints and inexact flows}
\label{sec:y-to-v}

Koster et al.~\cite[Lemma 4.1]{kloster2018practical} used an optimization to reduce the graph size in a pre-processing step by suitably contracting all edges entering nodes with in-degree 1, and respectively, exiting from nodes with out-degree 1, as we review below. This reduction has sometimes been used, under the name ``Y-to-V'', to simplify the input graphs to other problems, see e.g.~\cite{tomescu2017safe}. Note that the ILP solver by Dias et al.~\cite{dias2022fast} did not include pre-processing step. In this section, we show how to extent this pre-processing step for the two problem variants MIFD and MFDSC (missing proofs are in \Cref{apx:additional-proofs}).

The Y-to-V reduction for MFD works the following way. Consider the graph $(V,E)$ and let $v\in V$ be a node of in-degree $1$, and let $u\in V$ be the unique node with $(u,v) \in E$. Let $v$ have out-degree $\ell$ and let $w_i \in V$, $i=1,\dots,\ell$ be all nodes with $(v, w_i)\in E$. The reduction creates a new graph $(V,E')$ with edges $(u,w_i)$ and flow $f':E'\to\N$ with $f'(u,w_i) = f(v,w_i)$ for $i=1,\dots,\ell$. This process can be repeated until no nodes of in-degree $1$ exist anymore, while all other edges are copied to the new graph. Analogously this can be done for nodes of out-degree $1$.

For the MFDSC problem (with subpath constraints), note that subpaths are given with respect to the original graph, and hence we need to modify them for the Y-to-V contracted graph. The induced subgraphs of nodes of in-degree (resp. out-degree) $1$ define a forest of out-trees (resp. in-trees). Subpaths can potentially pass through, begin in or end in such trees.

If a path completely passes through a tree, we can merely merge the edges of the tree in the path to the contracted edge. In that case, the subpath which intersects with the tree gets translated a single new edge. If a path stops inside an out-tree or begins inside an in-tree, we can not do that because we do not know which leaf the path will cross in the decomposition. In that case, the subpath which intersects with the tree gets translated to several new, parallel edges.

We first contract all out-trees and describe how to reduce the path to the contracted graph. We can follow up by contracting all in-trees analogously. We can partition the graph uniquely in maximal-size out-trees, such that no adjacent out-trees can be merged to a new out-tree. An out-tree in this partition consists of a single edge if and only if this edge will not be contracted in the Y-to-V reduction. Let the subpath $R_j\in \subpaths$ intersect the out-trees $T_1,\dots,T_{\ell_j}$, ending at node $u\in T_{\ell_j}$. We can contract the path to obtain $R'_j$ in the following way. Let the set of leaf nodes of $T_{\ell_j}$ that can be reached from $u$ be denoted by $L(T_{\ell_j}, u)$. For every tree $T_1,\dots,T_{\ell_j-1}$, we add an edge from the root of the tree to the leaf the path is crossing to $R'_j$. For $T_{\ell_j}$ we add the parallel edges $(\text{root}(T_{\ell_j}), v)$ for all $v\in L(T_{\ell_j}, u)$. That means $R'_j$ is a path of length $\ell_j-1$ followed by some parallel edges. Let the obtained set of ``subpaths'' $R'_j$ be $\subpaths'$. We change the constraints \labelcref{eq:path-contains-subpath,eq:subpath-in-atleast-one-path} to:
\begin{align}
    \sum_{(u,v)\in R'_j} x_{uvi} \geq |\ell_j|r_{ij},\quad &\forall i\in\{1,\dots,k\}, \forall R'_j\in\subpaths', \label{eq:path-contains-subpath-contracted} \\
    \sum_{i\in\{1,\dots,k\}} r_{ij} \geq 1,\quad & \forall R'_j\in\subpaths'. \label{eq:subpath-in-atleast-one-path-contracted}
\end{align}
Note that only the set $\subpaths$ and the integers $|R_j|$ change compared to the constraints \labelcref{eq:path-contains-subpath,eq:subpath-in-atleast-one-path}.

\begin{lemma}
    The construction is correct, i.e. the constraints \labelcref{eq:path-contains-subpath-contracted,eq:subpath-in-atleast-one-path-contracted} are true in the original graph if and only if the constraints \labelcref{eq:path-contains-subpath,eq:subpath-in-atleast-one-path} are true in the Y-to-V contracted graph.
    \label{lem:y-to-v-subpaths}
\end{lemma}

The flow conservation is a necessary property for the Y-to-V reduction to work, since because of it, a single out-edge carries no further information. All paths of a flow decomposition entering a node of out-degree $1$ must continue through this edge, and will together decompose it (and similarly for nodes of in-degree $1$). In the case of MIFD, we do not require the input to have any conservation of flow, and unlike in MFD, valid inputs can also be infeasible to solve.

For nodes of in-degree $1$ (resp. out-degree $1$) we generalize flow conservation to the property of \emph{inexact flow conservation}, which nodes must have in order for its edges to be contracted:
\begin{equation}
    L_\text{in} \leq \sum_{(u,v)\in E} L_{uv} \leq \sum_{(u,v)\in E} R_{uv} \leq R_\text{in}, \label{eq:inexact-fc}
\end{equation}
where $[L_\text{in}, R_\text{in}]$ is the interval of the single in-edge of $u$ (respectively for the sum of all outgoing edges and the interval $[L_\text{out}, R_\text{out}]$ of the single out-edge of $u$). Note that the second inequality in \eqref{eq:inexact-fc} is always fulfilled since $L_{uv}\leq R_{uv}$. The Y-to-V reduction for MIFD works the same as the Y-to-V reduction for MFD, but defines intervals $[L'_{uw_i},R'_{uw_i}]$ (resp. $[L'_{w_iu},R'_{w_iu}]$) instead of flow values $f'(u,w_i)$ (resp. $f'(w_i,u)$) and is restricted on nodes of in-degree (resp. out-degree) $1$ fulfilling the inexact flow conservation.

\begin{lemma}
    The Y-to-V reduction for MIFD is correct.
    \label{lem:y-to-v-inexact}
\end{lemma}

\section{MFD with given weights}
\label{sec:wmfd}


In this section we consider an MFD variant in which the path weights in the solution are restricted to belong to a given set $W = \{w_1, \dots, w_\ell\}$, see also \Cref{fig:weighted} in \Cref{apx:additional-proofs} for an example and note that not all weights in $W$ must be used by the solution paths. This problem was also defined by Kloster et al.~\cite[Sec.~6]{kloster2018practical}, and in this paper we call it \emph{MFD with given weights (MFDW)}. More formally, given a flow network $G=(V,E,f)$ and a set of weights $W = \{w_1,\dots,w_\ell\} \subseteq \mathbb{Z}^+$, find a minimum-sized set of $s$-$t$ paths $\paths = (P_1,\ldots,P_k)$ with associated weights $(\tilde{w}_1,\ldots,\tilde{w}_k)$, with each $\tilde{w}_i \in W$, such that $\sum_{i \in \{1,\dots,k\}} P_i(u,v)\tilde{w}_i = f_{uv}$, $\forall (u,v) \in E$.

This problem has a smaller search space than MFD, because the weights are restricted to $W$, and thus can potentially admit faster solvers. Moreover, if the weights of an optimal solution are already known, then an optimal solution to \IA is also an optimal solution to MFD, and this problem can be used as potentially faster solver. For instance, Kloster et al.~\cite{kloster2018practical} have observed that path weights can be found in the flow in graphs used for RNA sequence reconstruction. 

In the rest of this section we show that knowing such a set $W$, we are able to formulate a model with substantially fewer variables than the ILP model for MFD, potentially decreasing the runtime of the ILP solver by a substantial amount. This formulation also has the potential to be an alternative efficient heuristic algorithm for MFD, if $W$ is well chosen.


Given the set of possible weights, our ILP model will able to answer the following question: In an optimal solution (that minimizes the number of used paths), how many paths of weight $w_i$ pass through an edge $e$? In other words, to decompose a flow $f$ using the weights in $W$, we aim to find flows $X_i$ such that $f = \sum_{i=1}^{|W|} w_iX_i$, where each $X_i(e)$ answers that question. The ILP model for MFDW is the following:
\begin{subequations}
\begin{align}
\allowdisplaybreaks
&\emph{Minimize} \quad k\coloneqq\sum_{(s,u)\in E} \sum_{i=1}^{|W|} x_{sui} \nonumber \\
&\emph{Subject to:} \nonumber \\
& f_{uv} = \sum_{i=1}^{|W|} w_i x_{uvi}, \quad && \forall (u,v)\in E,\\
& \sum_{(u,v)\in E} x_{uvi} - \sum_{(v,u)\in E} x_{vui} = 0, \quad &&\forall v\in V\setminus \{s,t\}, \forall i \in \{1,\dots,|W|\},\\
& x_{uvi} \in \mathbb{N}, \quad && \forall (u,v)\in E, \forall i \in \{1,\dots,|W|\}.
\end{align}
\end{subequations}

Note that the flow $X_i(u,v)=x_{uvi}$ can be decomposed trivially into weight $1$ paths. In addition, the product $x_{uvi}w_i$ does not require linearization due to $w_i$ being an input in this problem.

\begin{lemma} \label{lem:gwmfd_model_is_opt}
    The ILP model described in above solves MFDW optimally.
\end{lemma}

In \Cref{apx:additional-proofs} we further argue that MFDW is strongly NP-hard and that solving MFDW of $G$ with weights $\{2^i \mid i=0,\dots,\ceil{\log ||f||}\}$ optimally is a $\ceil{\log ||f||}+1$-approximation of MFD. Depending on the application, weight sets other than powers-of-two can give a better approximation or faster runtime. We can for example use powers of $2^i$ for any $i \in \mathbb{Z}^+$, which yields an approximation scheme with approximation factor $(2^i-1)\ceil{(\log||f||+1)/i}$. Increasing $i$, we expect a decreased runtime due to the reduced dimension of the ILP model, but the size of the decomposition increases.

\section{Experimental results}

\paragraph{Solvers and datasets.} 
Our implementation of the ILPs uses the \Gurobi Python API under default settings and is available at \url{https://github.com/algbio/optimized-fd}. We compare the ILPs for MFD, MFDSC and MIFD also to the heuristic algorithms for them by~\cite{shao2017theory} (\CF), by~\cite{williams2021flow} (\CSH) and by~\cite{williams2019rna} (\IFD), respectively. For \IA we used the weight set $W = \{ 2^i \mid i \leq \log ||f|| \} \cup \{ f(e) \mid e\in E \}$. 

We also experimented with the implementation from the \CF solver of the standard greedy algorithm for MFD~\cite{VATINLEN20081390}, but since we observed that it performs worse than the \CF heuristic algorithm, we did not include it in the results. We also did not include in the results the \TB implementation by Kloster et al.~\cite{kloster2018practical} of the FPT algorithm for MFD, nor the \CS implementation by Williams et al.~\cite{williams2021flow} of the FPT algorithm for MFDSC since it was already observed in~\cite{dias2022fast} that they do not scale for minimum flow decomposition sizes larger than $6$ (which we also confirmed experimentally).

The runtimes of all ILPs (except \IA) include a linear scan in increasing order to find the smallest $k$ for which there is a flow decomposition in $k$ paths. As discussed in \Cref{sec:optimized-ilp}, the size of the maximum weight antichain $Q$ is a lower bound on $k$, thus the linear scans for \OP, \SO and \IO start at $|Q|$. We set up a time limit of up to 30 minutes for each input and each method. Our experiments were performed in a server running Linux with one AMD Ryzen Threadripper PRO 3975WX 32-core CPU with 512 GB RAM.

\begin{sloppypar}
For the MFD problem, we experimented with the datasets created by Shao and Kingsford~\cite{shao2017theory}. This dataset was created from human transcriptome using the quantification tool Salmon~\cite{patro2015salmon} and also contains datasets that were created using Flux-Simulator~\cite{griebel2012modelling} from human and mouse transcriptomes. First, a sample file from their dataset referred as \emph{SRR020730~Salmon} (corresponding to the file \texttt{rnaseq/sparse\_quant\_SRR020730.graph}), followed by the entire archive present in the directories  \texttt{rnaseq/human/} and \texttt{rnaseq/mouse/}, referred as \emph{Catfish Human} and \emph{Catfish Mouse}, respectively. From those two directories, only the graphs that had 50 nodes (prior to Y-to-V) were considered. For MFDSC, two datasets were used: an adaptation of the \emph{SRR020730~Salmon}, where subpaths were generated and another dataset \emph{SRR30790~StringTie}, which was created by Khan et al.~\cite{Khan:2022wo} from human RNA-seq reads SRR307903 assembled using the StringTie tool~\cite{pertea2015stringtie}). In both datasets, a limited of four subpaths per test instance were considered (due to performance limitations regarding previous tools used as benchmarks). Finally, for MIFD, we simulate interval flows similar to what was done in~\cite{williams2019rna} and, as also described in~\cite{williams2019rna}, it is possible that infeasible instances are generated in the process. In this case, we repeated the process until only feasible instances were generated.
\end{sloppypar}

\paragraph{Experimental setup.}
\label{sec:discussion}
To show the behaviour of the solvers for graphs of increasingly large MFD, we group the input graphs in ranges based on their MFD size, computed with \OP. If \OP does not finish within the time limit on an instance, we exclude it also from all other solvers -- note that only the heuristic solvers are faster than \OP, which are not optimal in general. For each solver, we report the average runtime per range of graphs (column \textbf{Avg.}), and the total runtime in that range on the graphs (column \textbf{Total}), both in seconds, and only for the graphs on which that solver finished within the time limit. Thus, these numbers are an underestimation of the time one would need to run the solver in practice, since it would run for at least 30 minutes on unsolved instances. In column \textbf{\#Solved} we list the number of instances on which the solver finished. We captured the runtime with the GNU \texttt{time} tool by separating the graphs into individual files and running the tools separately on each instance.

Column \bm{$\Delta(|\paths|)$} shows the approximation accuracy of the heuristic methods as follows: for each instance, we compute the difference between the number of paths reported by each formulation and the minimum number of paths in an MFD (computed with \OP). In each table cell, corresponding to a specific range of inputs, we list the sum of these differences, and in parentheses their averages.  



\paragraph{Discussion.} \Cref{tab:FD_standard} illustrates the performance of the solvers for the MFD problem. \OP improves over \ST in all ranges of MFD size. We generally observe that the larger the MFD size, the more significant the runtime improvement. For example, for Catfish~Human ($\min k \geq 16$), \OP is 49$\times$ faster than \ST on average and for Catfish~Mouse ($\min k \geq 21$), \OP is 70$\times$ faster than \ST on average.

As mentioned above, in the reported runtimes of the solvers we are not including instances which did not finish within the time limit. Thus, in practice \OP has even larger speedups compared to \ST, since \OP solves more instances within the time limit than \ST, for Catfish~Human and Mouse this happening already for $\min k \geq 6$. In SRR020730~Salmon \OP solves all instances, and on the other two datasets \OP solves 22 and 46 more instances, respectively.

\begin{table*}[!h]
\centering
\caption{\footnotesize Results for Problem MFD. \ST denotes the original ILP for MFD from Dias et al.~\cite{dias2022fast} (\Cref{sec:basic}), \OP denotes our optimized ILP described in \Cref{sec:optimized-ilp} and \Cref{sec:y-to-v}, and \IA the ILP from \Cref{sec:wmfd}. Runtimes are in seconds; a timeout of 30 minutes was used. The total number of instances in the datasets is mentioned in the rows "All", in parentheses. Since the heuristic solvers finish on all instances, we do not have \#Solver columns for them. 
\label{tab:FD_standard}}
\resizebox{\columnwidth}{!}{%
\begin{tabular}{l r rrr rrr rrr rrr}
\toprule
\midrule
& &\multicolumn{3}{c}{\CF heuristic} &\multicolumn{3}{c}{\IA heuristic} &\multicolumn{3}{c}{\ST}& \multicolumn{3}{c}{\OP}\\
\cmidrule(r){3-5} \cmidrule(r){6-8} \cmidrule(r){9-11} \cmidrule(r){12-14}
& min $k$ & Avg. & Total  & $\Delta(|\paths|)$ & Avg. & Total & $\Delta(|\paths|)$ & Avg. & Total  & \#Solved & Avg. & Total  & \#Solved\\
\midrule
\multirow{5}{*}{\rotatebox{90}{\shortstack{\textbf{SRR020730}\\\textbf{Salmon}}}}
& 1-5     & 0.00 &  0.51 &  1 (0.00) & 0.29 &  11065.22  & 13 (0.00) &  0.30 &  11751.06 &  38703 &  0.28 &  11006.17 &  38703 \\
& 6-10    & 0.00 &  2.52 &  14 (0.01) & 0.31 &  621.26    & 2 (0.00) &   0.48 &  958.88 &  1988 &  0.30 &  602.23 &  1988 \\
& 11-15  &  0.01 &  0.91 & 6 (0.04) & 0.33 &  52.78     & 0 (0.00) &  2.02 &  327.71 &  162 & 0.37 &  60.28 &  162 \\
& 16-20   & 0.01 &  0.11 & 0 (0.00) & 0.38 &  4.89      & 0 (0.00) &   15.97 &  191.68 &  12 &  0.92 &  11.96 &  13 \\
& 21-\emph{max}    &  0.01 &  0.06 & 6 (1.50) &  0.39      &  1.57 & 0 (0.00) &  17.56 &  52.68 &  3 & 16.05 &  64.21 &  4 \\
\cmidrule(r){2-14}
& All  (40870) & 0.00 &  4.11 & 27 (1.55) & 0.29 & 11745.72 & 15 (0.00) &	0.32 &	13282.01	& 40868	& 0.29 & 11744.85 & 40870 \\
\midrule
\multirow{5}{*}{\rotatebox{90}{\shortstack{\textbf{Catfish}\\\textbf{Human}}}}
& 1-5     & 0.00 &  0.87 & 3 (0.00) &  0.30 &  3045.73 & 299 (0.03) &  0.35 &  3631.62 &  10301 &  0.29 &  2985.03 &  10301 \\
& 6-10    & 0.00 &  2.14 & 61 (0.06) & 0.34 & 367.38 & 243 (0.22) &  3.67 &  3976.98 &  1083 &  1.56 &  1692.51 &  1085 \\
& 11-15    & 0.01 &  1.27 & 52 (0.24) &  0.40 &  86.64 & 30 (0.14) &  16.94 &  3558.35 &  210 &  2.24 &  476.53 &  213 \\
& 16-20    &  0.01 &  1.20 & 132 (1.08)  &  0.38 &  46.38 & 6 (0.05) &  43.56 &  4747.77 &  109 &  0.77 &  94.28 &  122 \\
& 21-\emph{max}   & 0.01 &  0.09 & 16 (2.67) &  0.40 &  2.38 & 0 (0.00) &  50.81 &  101.63 &  2 &  0.70 &  4.19 &  6 \\
\cmidrule(r){2-14}
& All (11730) & 0.00	& 5.57	& 264 (0.02) & 0.30	& 3548.51	& 578 (0.05) & 1.37 & 16016.35	& 11705 & 0.45 & 5252.54 & 11727 \\ 
\midrule
\multirow{5}{*}{\rotatebox{90}{\shortstack{\textbf{Catfish}\\\textbf{Mouse}}}}
& 1-5 &     0.00 &  0.51 & 2 (0.00) &  0.29 &  3532.97 & 280 (0.02) &  0.35 &  4202.92 &  12047 &  0.29 &  3471.02 &  12047 \\
& 6-10 &    0.00 &  5.67 & 139 (0.09) &  0.35 &  558.86 & 577 (0.36) &  3.29 &  5247.45 &  1594 &  1.16 &  1847.76 &  1597 \\
& 11-15 &   0.01 &  3.07 & 255 (0.78) &  0.55 &  180.58 & 171 (0.52) &  72.22 &  21234.09 &  294 &  23.08 &  7339.90 &  318 \\
& 16-20 &   0.01 &  0.74 & 69 (0.93) &  0.43 &  31.64 & 6 (0.08) &  58.18 &  3607.16 &  62 & 12.38 &  916.01 &  74 \\
& 21-\emph{max} &  0.01 &  0.44 & 57 (1.84) &  0.44 &  13.71 & 0 (0.00)  &  77.59 &  1862.17 &  24 &  0.85 &  26.43 &  31 \\
\cmidrule(r){2-14}
& All (14079) &  0.00	& 10.43	& 522 (0.04) &	0.31 &	4317.76	& 1034 (0.07) &	2.58 & 36153.79 &	14021 & 0.97 &	13601.12 & 14067 \\
\midrule
\bottomrule
\end{tabular}
}
\end{table*}

Among the two heuristic solvers, \CF finishes on all instances in a few seconds in total, whereas \IA is slower than \CF, but still running in under one second on average per graph and faster than \OP in almost all cases. When comparing the approximation accuracy, note that the instances of practical interest are those for which the MFD size is more than 10, since for smaller sizes \OP already returns an optimal solution seconds in under 2 on average. On such instances, \IA is more accurate than \CF, and interestingly, the larger the MFD size, the more accurate \OP becomes (for $\min k \geq 21$ being fully accurate on all datasets). One reason might be that in more complex graphs, the weights of an optimal MFD  appear among the flow values of the edges, which we added to the set $W$.



In \Cref{tab:MFDSC_results}, we show the results for the MFDSC problem. We observe a similar and even more pronounced trend as for the MFD solvers, where the larger the MFDSC size, the larger the improvement of our optimized ILP. For example, for dataset SRR020730~Salmon, \SO is 11$\times$, 58$\times$ and 54$\times$ faster than \SP for $\min k$ in ranges 11-15, 16-20, 21-\emph{max}, respectively. For both datasets, \SO manages to solve one more instance than \SP, and in SRR020730~Salmon it solves all instances. Moreover, the total running time of \CSH is more than half of the running time of \SO, while not giving optimal solutions (which are also less accurate on average than the heuristic solvers for MFD).

\begin{table*}[!h]
\centering
\caption{\footnotesize Results for Problem MFDSC. Runtimes are in seconds. A timeout of 30 minutes was used. \SP is the original ILP from from Dias et al.~\cite{dias2022fast}, and \SO is our optimized ILP as described in \Cref{sec:optimized-ilp} and \Cref{sec:y-to-v}. Note that the dataset SRR30790
StringTie has no instances where the optimum MFDSC solution size is larger than 15. 
\label{tab:MFDSC_results}}
\resizebox{0.8\columnwidth}{!}{%
\begin{tabular}{l r rrr rrr rrr}
\toprule
\midrule
&&\multicolumn{3}{c}{\CSH}& \multicolumn{3}{c}{\SP} &\multicolumn{3}{c}{\SO}\\
\cmidrule(r){3-5} \cmidrule(r){6-8} \cmidrule(r){9-11}
& min $k$ & Avg. & Total  & $\Delta(|\paths|)$ & Avg. & Total  & \#Solved & Avg. & Total & \#Solved \\
\midrule
\multirow{5}{*}{\rotatebox{90}{\shortstack{\textbf{SRR020730}\\\textbf{Salmon}}}}
& 1-5 & 0.19 &  676.67 &  109 (0.03) & 0.70 &  2543.97 &  3643 &  0.29 &  1057.00 &  3643 \\
& 6-10 & 0.19 &  377.83 &  411 (0.21) & 1.11 &  2214.29 &  1988 &   0.30 &  593.56 &  1988 \\
& 11-15 & 0.20 &  32.15 &  126 (0.78) & 3.79 &  613.24 &  162 &   0.34 &  55.41 &  162 \\
& 16-20  &  0.21 &   2.72 &  22 (1.69) &  70.52 &  846.29 &  12 & 1.12 &  14.54 &  13 \\
& 21-\emph{max} & 0.22 &  0.88 &  5 (1.25) & 93.28 &  373.12 &  4 &   1.71 &  6.82 &  4 \\
\cmidrule(r){2-11}
& All (5810) &  0.19 &  1090.25 & 673 (0.12) & 1.13 & 6590.91 & 5809 &  0.29 & 1727.33 & 5810 \\
\midrule
\multirow{5}{*}{\rotatebox{90}{\shortstack{\textbf{SRR30790}\\\textbf{StringTie}}}}
& 1-5 & 0.18 &  152.19 &  190 (0.23) &  0.69 &  568.43 &  828 & 0.31 &  255.99 &  828 \\
& 5-10 & 0.19 &  43.14 &  75 (0.32) & 1.24 &  287.75 &  232 & 0.33 &  76.59 &  232 \\
& 11-15 & 0.20 &  1.19 &  1 (0.17) & 2.31 &  11.56 &  5 &  0.90 &  5.42 &  6 \\
\cmidrule(r){2-11}
& All (1067) & 0.18	& 196.52	& 266 (0.25) & 0.81 & 867.74 & 1065 & 0.32	& 338.00 & 1066\\
\midrule
\bottomrule
\end{tabular}
}
\end{table*}

In \cref{tab:MIFD_results}, we show the results for MIFD problem and again observe a significant speedup compared to \IX on complex instances. For example, for $\min k$ in 11-15, \IO is up to 15$\times$ faster on SRR020730 Salmon, up to 123$\times$ faster on Catfish Human and up to 47$\times$ faster on Catfish Mouse. On the latter two datasets, \IO drastically runs in less than $1$ second on average. In addition, although \IFD has, on average, a better runtime, its difference compared to \IO is negligible.



\begin{table*}[!h]
\centering
\caption{\footnotesize Results for Problem MIFD. Runtimes are in seconds. A timeout of 30 minutes was used. \IX is the original ILP from from Dias et al.~\cite{dias2022fast}, and \IO is our optimized ILP as described in \Cref{sec:optimized-ilp} and \Cref{sec:y-to-v}. Note that the dataset Catfish Human has no instances where the optimum MIFD solution size is larger than 15 and the dataset Catfish Mouse has no instances where the optimum MIFD solution size is larger than 20.
\label{tab:MIFD_results}}
\resizebox{0.8\columnwidth}{!}{%
\begin{tabular}{l r rrr rrr rrr}
\toprule
\midrule
&&\multicolumn{3}{c}{\IFD heuristic}& \multicolumn{3}{c}{\IX} &\multicolumn{3}{c}{\IO}\\
\cmidrule(r){3-5} \cmidrule(r){6-8} \cmidrule(r){9-11}
& min $k$ & Avg. & Total  & $\Delta(|\paths|)$ & Avg. & Total  & \#Solved & Avg. & Total & \#Solved \\
\midrule
\multirow{5}{*}{\rotatebox{90}{\shortstack{\textbf{SRR020730}\\\textbf{Salmon}}}}
& 1-5 &  0.25 &  9561.44 &  2161 (0.06) &  0.58 &  22559.20 &  38706 &  0.57 &  22165.24 &  38706 \\
& 6-10 &  0.26 &  512.92 &  1697 (0.85) &  1.03 &  2044.12 &  1986 &  0.61 &  1219.30 &  1986 \\
& 11-15 &  0.29 &  46.58 &  302 (1.88) &  11.72 &  1887.70 &  161 &  0.70 &  112.70 &  161 \\
& 16-20 & 0.32 &  4.12 &  46 (3.54) &  28.18 &  309.97 &  11 &  3.25 &  42.21 &  13 \\
& 21-\emph{max} &  0.44 &  1.74 &  21 (5.25) &  66.30 &  198.90 &  3 &  10.94 &  43.75 &  4 \\
\cmidrule(r){2-11}
& All (40870) & 0.25 & 10126.80 & 4227 (0.10) & 0.66 & 26999.89 & 40867 & 0.58 & 23583.20 & 40870   \\
\midrule
\multirow{5}{*}{\rotatebox{90}{\shortstack{\textbf{Catfish}\\\textbf{Human}}}}
& 1-5 &  0.24 &  2652.76 &  957 (0.09) &   0.59 &  6398.51 &  10911 &  0.58 &  6318.75 &  10911 \\
& 6-10 &  0.26 &  40.52 &  198 (1.25) &  1.36 &  215.40 &  158 &  0.65 &  101.99 &  158 \\
& 11-15 &  0.24 &  0.98 &  14 (3.50) &  101.01 &  404.03 &  4 &  0.81 &  3.25 &  4 \\
\cmidrule(r){2-11}
& All (11073) & 0.24 & 2694.26 & 1169 (0.11) & 0.63 & 7017.94 & 11073 & 0.58 & 6423.99 & 11073 \\
\midrule
\multirow{5}{*}{\rotatebox{90}{\shortstack{\textbf{Catfish}\\\textbf{Mouse}}}}
& 1-5 &  0.24 &  3182.64 &  763 (0.06) &  0.58 &  7562.22 &  13009 &  0.58 &  7519.09 &  13009 \\
& 6-10 &   0.25 &  26.83 &  139 (1.31) &  1.47 &  155.55 &  106 &  0.63 &  67.15 &  106 \\
& 11-15 & 0.36 &  1.79 &  19 (3.80) &  45.35 &  226.77 &  5 &  0.94 &  4.72 &  5 \\
& 16-20 &  0.29 &  0.29 &  3 (3.00) &  0.00 &  0.00 &  0 &  0.91 &  0.91 &  1 \\
\cmidrule(r){2-11}
& All (13121) & 0.24 & 3211.55 & 924 (0.07) & 0.61 & 7944.54 & 13120 & 0.58 & 7591.87 & 13121 \\
\midrule
\bottomrule
\end{tabular}
}
\end{table*}

\section{Conclusion}
\label{sec:conclusion}


In this paper, we proposed optimizing the ILP formulations for MFD using the notion of safe paths for all flow decompositions. Since safe paths cannot be simply be removed from the graph, we observed that we can use a set of independent safe path to suitably fix a large number of ILP variables corresponding to their edges. Combined with the Y-to-V reduction, and the first usage of an antichain lower bound in a solver to detect the optimality of a heuristic solution, this resulted in a significantly faster MFD solver (up to 91$\times$ on the harder instances). We also developed an ILP that can work as a heuristic for MFD, running in under 1 second on average, and more accurate than the state-of-the-art heuristic \CF on graphs of practical interest. We also showed that these optimizations can be applied to two MFD variants of practical interest (by also adapting the Y-to-V reduction the first time for these variants), resulting in even bigger speedups.


Future research encompasses extending such improvements to the ILPs for MFD in general graphs with cycles~\cite{dias2022minimum}. Another point of improvement is a further search space reduction through further optimizations on the application of safety, for example, by answering the following question: Is it tractably possible to find subpaths that are part of all $\alpha$-approximations of MFD, where $\alpha$ can either be a constant or depend on the input size? Such \emph{$\alpha$-safe paths} (following the notion of $c$-essential vertices by Bumpus et al.~\cite{bumpus2022search}) would be longer the closer $\alpha$ is to $1$, while safe paths in this paper would refer to \emph{$\infty$-safe paths}.
For the problem variants MFDSC and MIFD, can we find all $\infty$-safe paths in polynomial time?

\bibliography{main}
\bibliographystyle{plain}

\clearpage
\appendix

\newpage

\section{Additional proofs and figures}
\label{apx:additional-proofs}

This section contains the proofs and figures missing from the main matter, and some additional results for the MFDW problem.

\begin{figure*}[h]
\centering
\begin{subfigure}[c]{0.45\linewidth}
\centering
\includegraphics[width=0.65\textwidth]{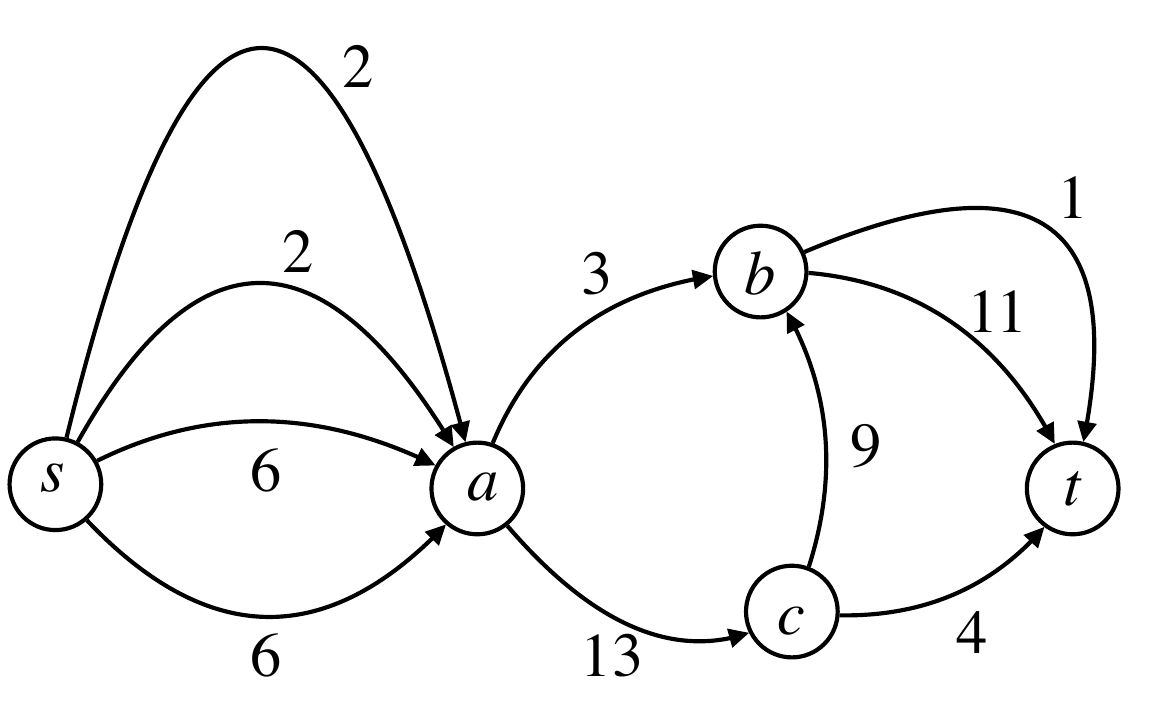}
\caption{A flow network \label{fig:og}}
\end{subfigure}
\begin{subfigure}[c]{0.45\linewidth}
\centering
\includegraphics[width=0.65\textwidth]{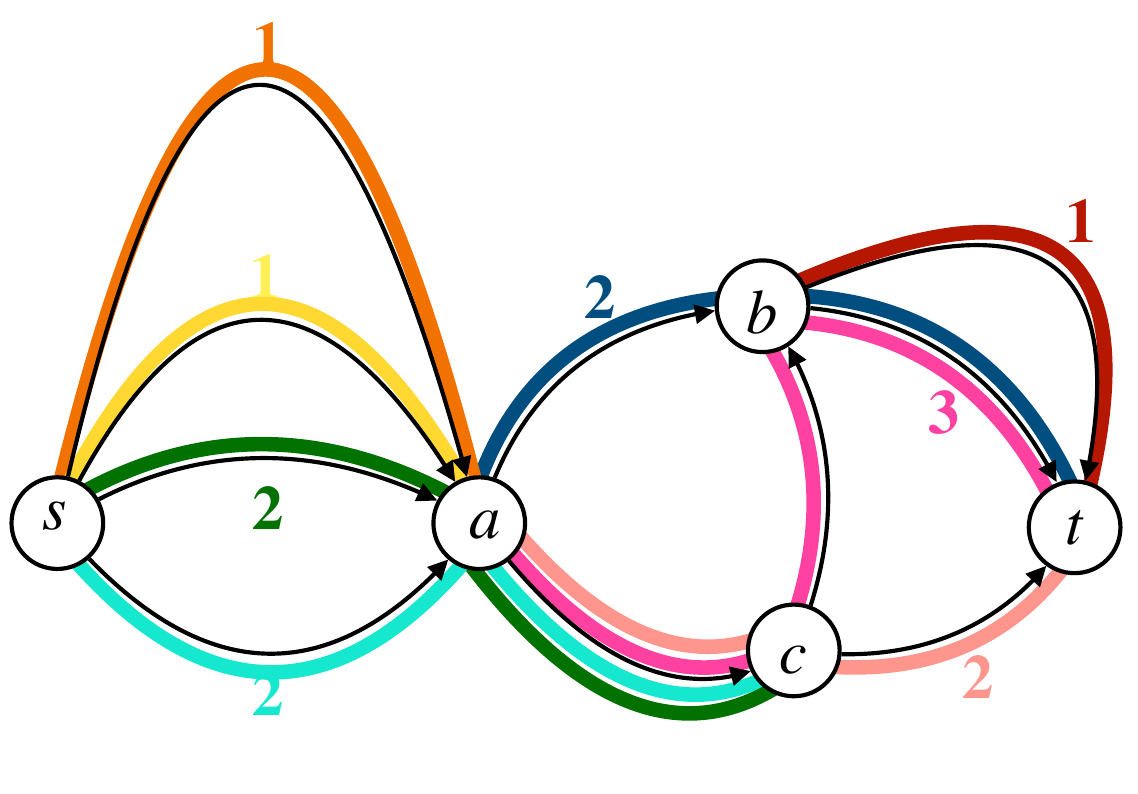}
\caption{All maximal safe paths with their respective lengths.
\label{fig:safe}}
\end{subfigure}\\
\begin{subfigure}[c]{0.45\linewidth}
\centering
\includegraphics[width=0.65\textwidth]{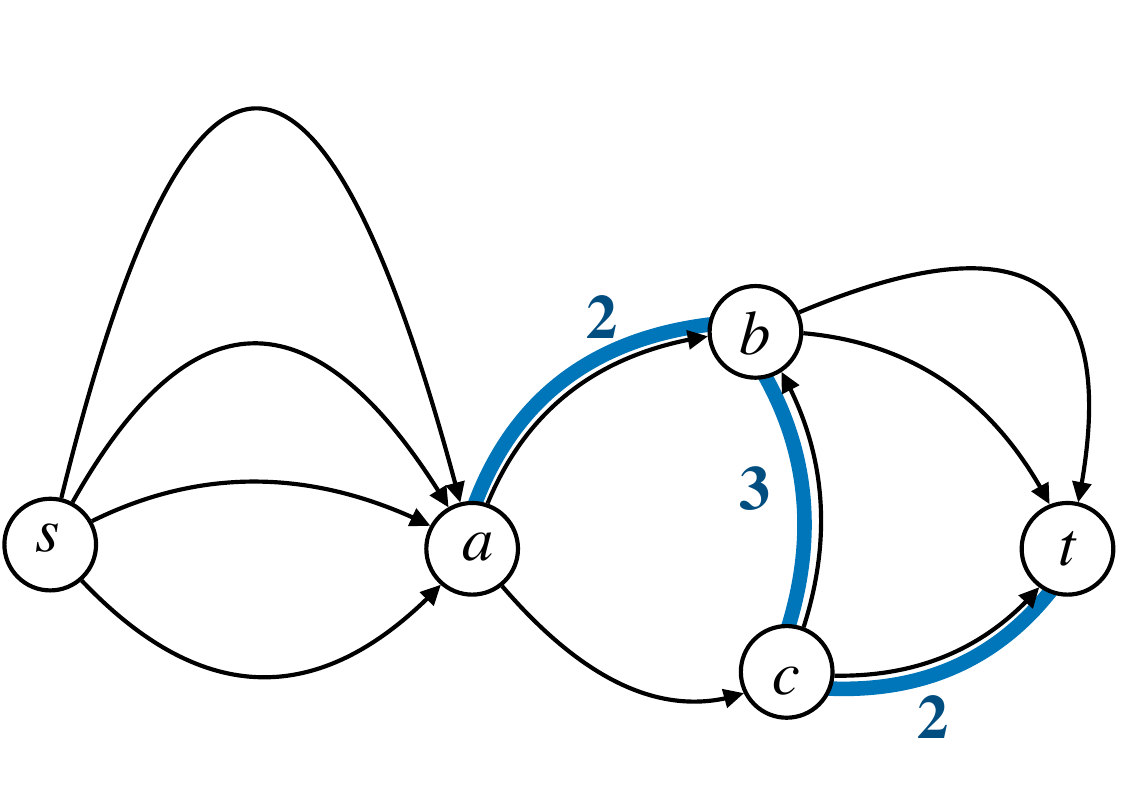}
\caption{Maximum Weight Antichain.\label{fig:anti}}
\end{subfigure}
\begin{subfigure}[c]{0.45\linewidth}
\centering
\includegraphics[width=0.65\textwidth]{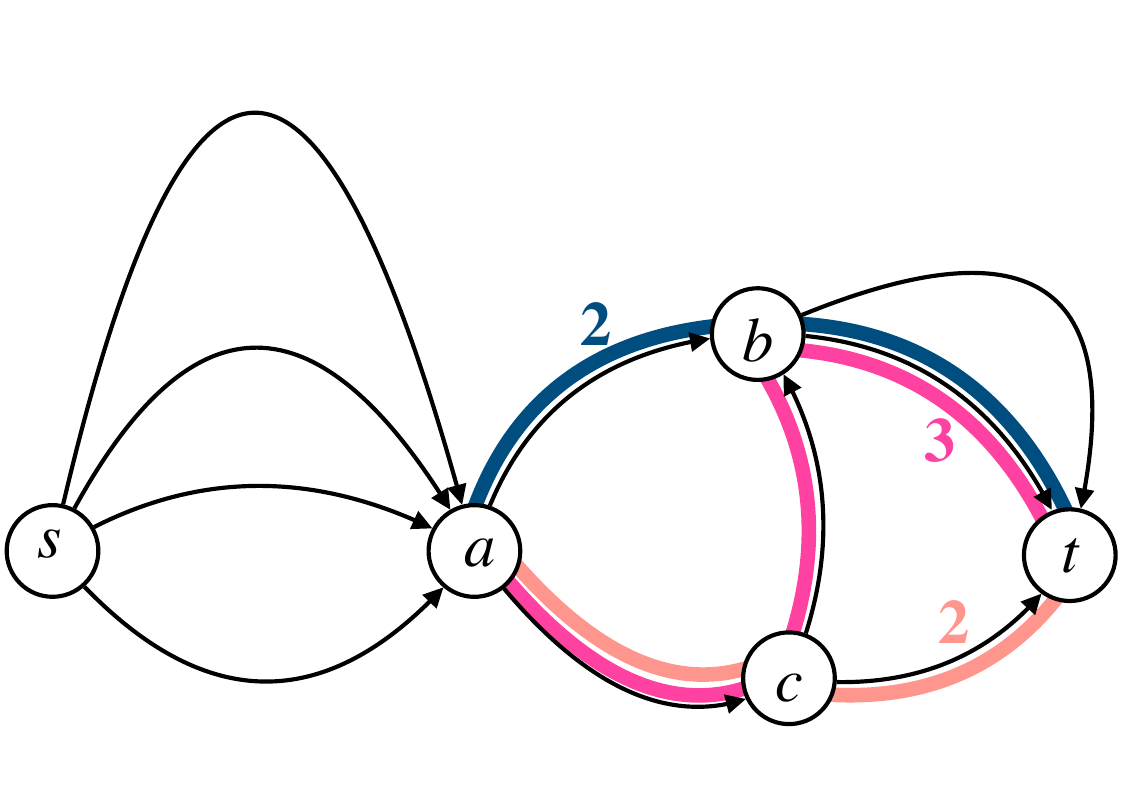}
\caption{Long collection of independent safe paths.
\label{fig:chosen}}
\end{subfigure}
\caption{Application of safe paths in introducing new constraints into previous ILP models. A flow network is displayed in~\Cref{fig:og}. From this network, all maximal safe paths (that cannot be extended left or right) can be calculated (displayed in different colours with their respective lengths in~\Cref{fig:safe}). We attach to every edge the maximum length of all safe paths crossing that edge as weight. By calculating the maximum weight antichain~(\Cref{fig:anti}), we obtain a set of pairwise independent safe paths~(\Cref{fig:chosen}). In the ILP formulation, we set to $1$ the $x_{uvi}$ variables of each edge $(u,v)$ in the $i$-th safe path of the maximum weight antichain.\label{fig:safe-paths}}
\end{figure*}

\begin{figure}[ht]
\centering
\includegraphics[width=0.33\linewidth]{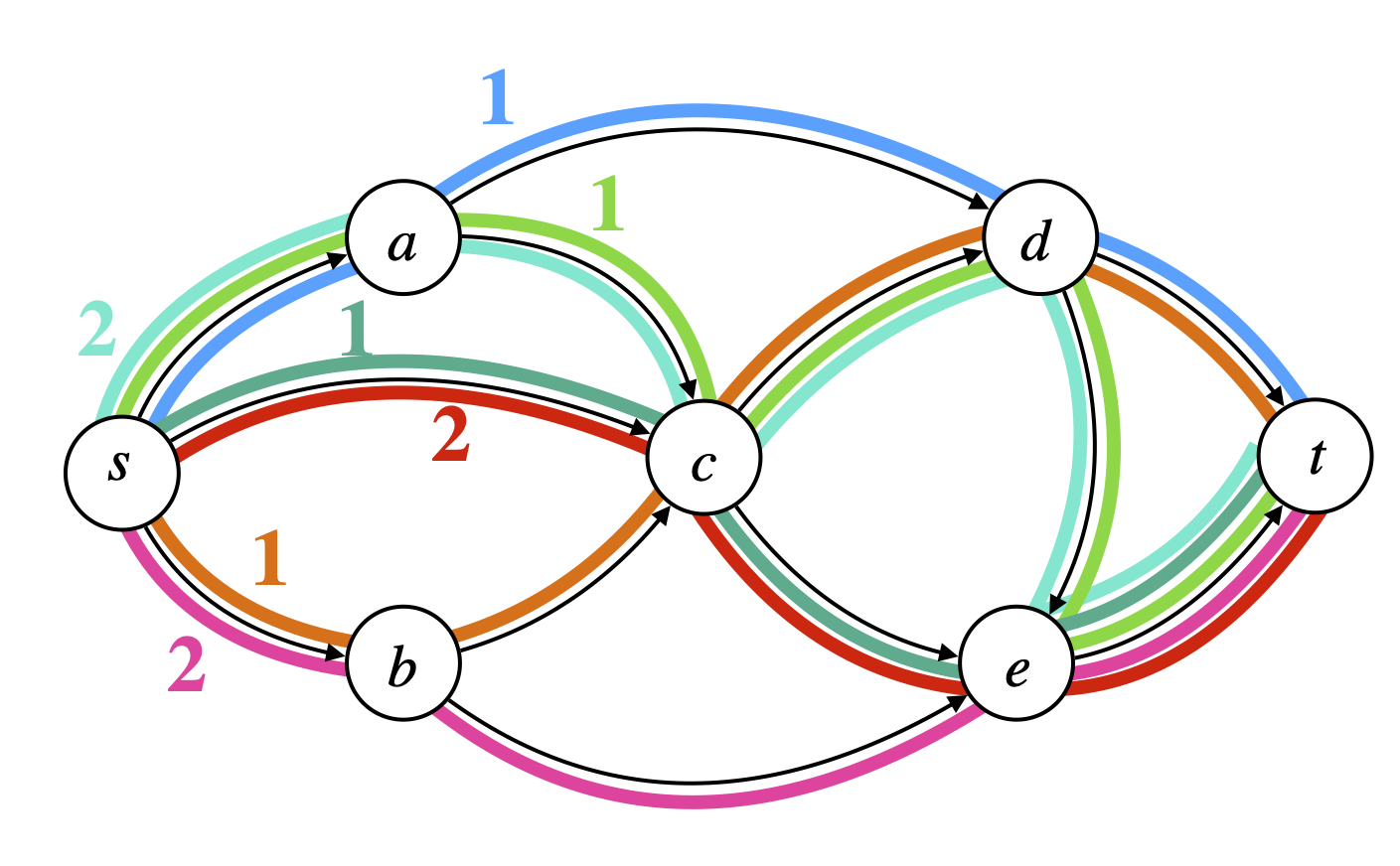}
\caption{Another flow decomposition of the flow network displayed in Fig.~\ref{fig:MFD}. In this scenario, all weights are already known and limited to only power of 2 (1, 2, 4, 8) values.\label{fig:weighted}}
\end{figure}

\begin{proof}[Proof of \Cref{lem:inexact-excess-flow}]
    Let $f_P > 0$ and assume that no path of a flow decomposition $\paths$ of assigned weights $w\in{\mathbb{Z}^+}^{|\paths|}$ routed through edge $e_1$ passes through all of $P$. This means they all leave $P$ through edges $\text{OUT}\coloneqq \{ (u_i, v)\in E \mid i\in\{2,\dots,k\}, v\neq u_{i+1} \}$ carrying weight at most $\sum_{e\in\text{OUT}} R_e < L_{e_1}$. This contradicts the assumption that $\sum_{P_j\in\paths:P_j(e_1)=1} w_j \in [L_{e_1}, R_{e_1}]$.
\end{proof}

\begin{proof}[Proof of \Cref{lem:y-to-v-subpaths}]
    We analyze the possible routes a path $R_j$ can be extended to, such that it begins and stops at nodes that lie in between different out-trees of the graph partition.

    The path $R_j$ begins in the out-tree $T_1$, potentially not at the root. Since $T_1$ is an out-tree, it can uniquely be extended to the root w.l.o.g., and we can assume that it starts at the root of $T_1$.

    As such, the path passes through the out-trees $T_1,\dots,T_{\ell_j-1}$ from the out-tree root to one of its leaves, which will be the root of the next tree. For these trees, there exists a unqiue edge in the contracted graph, that connects the root to that leaf.

    We now analyze the path reduction for the last out-tree $T_\ell$. $R'_j$ has been constructed to contain $\ell_j-1$ edges that form a path in the Y-to-V contracted graph, followed by parallel edges adjacent to that path. These parallel edges describe exactly all possibilities for $R_j$ to be extended to the right, such that it ends outside of the out-tree. Since they are parallel, and the variables $x_{uvi}$ are constrained to represent paths, exactly one of the parallel edges is forced to be used by constraint \eqref{eq:path-contains-subpath-contracted}.
\end{proof}

\begin{proof}[Proof of \Cref{lem:y-to-v-inexact}]
    Consider an inexact flow decomposition of the original graph. By uniquely adapting the paths to the contracted graph, they trivially define a feasible solution, as the same set of weighted paths contribute to the same intervals as before.
    
    Consider an inexact flow decomposition of the contracted graph. By uniquely adapting the paths to the original graph, they define a feasible solution, since the inexact flow conservation enforces the sum of the weights of paths passing through an edge with an interval that has been removed in the contracted graph to lie inside this interval.
\end{proof}

\begin{proof}[Proof of \Cref{lem:gwmfd_model_is_opt}]
    Let $\paths = \{P_1,\dots,P_k\}$ with weights in $W=\{w_1,\dots,w_\ell\}$ be a flow decomposition. Let $\paths(w_i) \subseteq \paths$ denote the multiset of paths of weight $w_i\in W$ and let $x_{uvi} = |\paths(w_i) \cap \{ \text{$s$-$t$ paths crossing $(u,v)$} \}|$. The variables $x_{uvi}$ define a feasible solution to the ILP model and we have $\sum\limits_{(s,u)\in E} \sum_{i=1}^{|W|} x_{sui} = k$ for the objective function.
    
    Given a feasible solution of the ILP model consisting of variables $x_{uvi}$, construct $s$-$t$ paths $P$ of weight $w_i\in W$ in the following way. Start at node $s$, choose an arbitrary neighbour $u$ of $s$ with $x_{sui} > 0$ and append the edge $(s,u)$ to $P$. Reduce $x_{sui}$ by one, and continue this step with $u$ instead of $s$ until the current node is $t$, which finishes the path $P$. Repeating this process until all the $x_{uvi}$ are zero yields a flow decomposition using $\sum\limits_{(s,u)\in E} \sum_{i=1}^{|W|} x_{sui}$ many paths.
\end{proof}


\begin{lemma} \label{lem:w-superset-opt}
    Let $W = \{w_1,\dots,w_\ell\}$ be a set of weights, and $\paths$ with associated weights $W' = \{w'_1,\dots,w'_k\}$ be an optimal solution of an MFD instance $G = (V,E,f)$. If $W' \subseteq W$, then solving \IA with weight set $W$ optimally solves MFD optimally.
\end{lemma}
\begin{proof}
    This follows from the fact that an optimal solution of MFDW always contains at least as many paths as an optimal solution of MFD.
\end{proof}

\begin{corollary}
    \IA is strongly NP-hard.
\end{corollary}
\begin{proof}
    Hartman et al. have shown that MFD is NP-hard on graphs using flow values only from $\{1,2,4\}$~\cite{hartman2012split}. \IA with $W = \{1,\dots,c\}$ solves MFD on instances where the flows are upper bounded by a constant $c$ by \Cref{lem:w-superset-opt}.
\end{proof}

When dealing with instances in which there is no known small set $W$ of which the path weights of an optimal solution of MFD come from, we can also use \IA to approximate the solution of MFD.

\begin{lemma}
    For an MFD instance $G = (V,E,f)$, solving \IA of $G$ with weights $\{2^j \mid j=0,\dots,\ceil{\log ||f||}\}$ optimally is a $\ceil{\log ||f||}+1$-approximation of MFD.

    More generally, solving \IA of $G$ with weights $\{{(2^i)}^j \mid j=0,\dots,\ceil{\log||f||/i}$ for some $i\in\mathbb{Z}^+$ is a $(2^i-1)\ceil{(\log||f||+1)/i}$-approximation of MFD.
\end{lemma}
\begin{proof}
Let $\paths = \{P_1,\dots,P_k\}$ with associated weights $\{w_1,\dots,w_k\}$ be an optimal solution to $G$ for MFD and let $\paths_w$ be an optimal solution for \IA. We can construct a feasible solution for \IA using $\paths$ by copying every $P_i \in \paths$ for every positive $j$-th bit in the power of two decomposition of $w_i$ and assigning weight $2^j$ to that copy. Since $w_i \leq ||f||$, this yields at most $\ceil{\log||f||}+1$ paths for every $P_i \in \paths$. Thus, $|\paths_w|\leq k\ceil{\log||f||+1}$.

If we only use every $i$-th bit, the bits on positions divisible by $i$ need to cover the next $i-1$ bits. Let the path $P$ have weight $w = \sum_{j=0,\dots,\ceil{\log||f||}} b_j2^j$ for $b_j\in\{0,1\}$, then $w_i = \sum_{j=0,\dots,\ceil{\log||f||/i}} 2^{ij}(b_{ij} + 2b_{ij+1} + \dots + 2^{i-1}b_{ij+i-1}) =  2^0(b_0+2b_1+\dots+2^{i-1}b_{i-1}) + 2^i(b_i+2b_{i+1}+\dots+2^{i-1}b_{2i-1}) + \dots$.
The sums $(b_{ij} + 2b_{ij+1} + \dots + 2^{i-1}b_{ij+i-1}) \leq 2^i-1$ denote how often we need to copy $P$ using weight $2^{ij}$.
\end{proof}

\section{ILP formulations}
\label{sec:flowdecomp}

\allowdisplaybreaks
\subsection{\ST}
\label{apx:standard-ilp}
\begin{subequations}
\label{eq:mod:ST}
\begin{align}
& {\forall i \in \{1, \ldots, k\}:}  \nonumber \\
& \quad \sum_{(s,v) \in E} x_{svi} = 1 \\
& \quad \sum_{(u,t) \in E} x_{uti} = 1  \\
& \quad \sum_{(u,v) \in E} x_{uvi} - \sum_{(v,w) \in E} x_{vwi} = 0 &&  \forall v \in V \setminus \{s, t\},\\
& f_{uv} = \sum_{i \in \{1,\dots,k\}} \phi_{uvi} && \forall (u,v) \in E, \\
& {\forall (u,v) \in E, \forall i \in \{1,\dots,k\}:}  \nonumber \\
& \quad \phi_{uvi} \leq f_\text{max} x_{uvi}\\
& \quad \phi_{uvi} \leq w_i \\
& \quad \phi_{uvi} \geq w_i - (1-x_{uvi})f_\text{max} \\
& w_i \in \mathbb{Z}^+ && \forall i \in \{1, \ldots, k\}, \\
& x_{uvi} \in \{0,1\} && \forall (u,v) \in E, \forall i \in \{1,\dots,k\},\\
& \phi_{uvi} \in \N && \forall (u,v) \in E, \forall i \in \{1,\dots,k\}.
\end{align}
\end{subequations}

\subsection{\OP}
\label{apx:optimized-ilp}
\text{Equations \eqref{eq:mod:ST} together with:} 
\begin{subequations}
\allowdisplaybreaks
\label{mod:OP}
\begin{align}
x_{uvi} = 1, \quad &&\forall (u,v)\in P_i, \forall P_i\in \paths(Q) \\
w_{i+1} \leq w_i, \quad && \forall i \in \{ |Q|+1,\dots,k-1\}
\end{align}
\end{subequations}

\subsection{\IA}
\begin{subequations}
\allowdisplaybreaks
\label{mod:IA}
\begin{align}
& \emph{Minimize}\quad \sum_{(s,u)\in E} \sum_{i=1}^{|W|} x_{sui}\nonumber\\
& \emph{Subject to:} \nonumber \\
& f_{uv} = \sum_{i=1}^{|W|} w_i x_{uvi}, \quad && \forall (u,v)\in E\\
& \sum_{(u,v)\in E} x_{uvi} - \sum_{(v,u)\in E} x_{vui} = 0, \quad && \forall v\in V\setminus \{s, t\}, \forall i \in \{1,\dots,|W|\} \\
& x_{uvi} \in \mathbb{N}, \quad && \forall (u,v)\in E, \forall i \in \{1,\dots,|W|\}
\end{align}
\end{subequations}

\section{Maximum length independent safe paths}
\label{apx:optimal-safe-paths}

The approach used in \Cref{sec:optimized-ilp} to calculate independent safe paths of a flow network $(G,f)$ does not return them of maximum length, but works only as a heuristic. In order to find independent safe paths of maximum length in polynomial time, one can, as in the heuristic approach, reduce to maximum weight antichains, on the following \textit{dependency} graph $\mathcal{D}=(V(\mathcal{D}),E(\mathcal{D}))$: nodes represent safe paths, with given weights of their length. A directed edge $(u,v)$ is added if there exists a path in $G$ which first traverses both safe paths represented by $u$ and $v$, entering the safe path of $u$ first and then entering the safe path of $v$. 

\begin{lemma}
    $A\subseteq V(\mathcal{D})$ is an independent set in $\mathcal{D}$ if and only if the set of safe paths corresponding to the nodes in $A$ is independent in $G$, where the total weight of $A$ is the total length of the corresponding safe paths. As a result, a maximum length independent path set of safe paths in $G$ can be calculated in polynomial time $O(|\mathcal{S}|^2\cdot\sum_{S\in\mathcal{S}} |S|)\subseteq O(|\mathcal{S}|^2\cdot n\cdot m)$, where $\mathcal{S}$ is the set of all safe paths.
\end{lemma}
\begin{proof}
    The graph $\mathcal{D}$ is a transitive graph, i.e. $(u,v),(v,w)\in E(\mathcal{D})$ implies that $(u,w)\in E(\mathcal{D})$. That is because path dependencies are transitive: If there is a path $P_1$ traversing safe paths $s_1$ and $s_2$, and there is a path traversing safe paths $s_2$ and $s_3$, one can construct a path that traverses all of $s_1,s_2$ and $s_3$. Thus, independent node sets in $\mathcal{D}$ correspond exactly to independent safe paths in $G$: Two safe paths are independent if and only if they are not connected by an edge in $\mathcal{D}$. The graph $\mathcal{D}$ can be constructed in $O(|\mathcal{S}|^2\cdot m)$ time, by performing a graph search for every pair of safe paths to check whether they are independent or not.

    It is a classical result~(\cite{mohring1985algorithmic}), that finding a maximum weighted independent node set on transitive graphs is solvable in polynomial time by finding a maximum weight node antichain. Since the weights of the nodes are $|S|$ for $S\in\mathcal{S}$, it runs in time $O(|E(\mathcal{D})|\cdot \sum_{S\in\mathcal{S}}|S|) \subseteq O(|\mathcal{S}|^2\cdot\sum_{S\in\mathcal{S}}|S|)$.
\end{proof}

%

\end{document}